\documentclass[a4paper, 12pt]{article}
\usepackage{amsmath}
\usepackage{amssymb}
\usepackage{graphicx}
\usepackage{multirow}
\usepackage{float}
\usepackage{mathtools}
\usepackage{dsfont}
\usepackage{placeins}
\title{A Nonparametric Statistical Approach to Content Analysis of Items} 

\newtheorem{proposition}{Proposition}

\newenvironment{proof}{\paragraph{Proof:}}{\hfill$\square$}

\author{%
	Diego Marcondes\footnote{Universidade de S\~ao Paulo, Brazil (dmarcondes@ime.usp.br)}
	\and 
	Nilton Rogerio Marcondes\footnote{Universidade de Coimbra, Portugal}}
\providecommand{\keywords}[1]{\textbf{\textit{Keywords---}} #1}
\date{}

\usepackage{filecontents}

\begin{document}
	
	\maketitle

\begin{abstract}
In order to use psychometric instruments to assess a multidimensional construct, we may decompose it in dimensions and, in order to assess each dimension, develop a set of items, so one may assess the construct as a whole, by assessing its dimensions. In this scenario, the content analysis of items aims to verify if the developed items are assessing the dimension they are supposed to. In the content analysis process, it is customary to request the judgement of specialists in the studied construct about the dimension that the developed items assess, what makes it a subjective process as it relies upon the personal opinion of the specialists. This paper aims to develop a nonparametric statistical approach to the content analysis of items in order to present a practical method to assess the consistency of the content analysis process, by the development of a statistical test that seeks to determine if all the specialists have the same capability to judge the items. A simulation study is conducted to assess the consistency of the test and it is applied to a real validation process. \\
\end{abstract}

\keywords{Nonparametric statistics; Applied statistics; Content validity; Psychometric instruments; Psychometrics.}

\section{Introduction}

Psychometric instruments play an important role in researches in the areas of psychology and education, thus it is necessary that they are thoroughly developed and validated, so that no erroneous results are obtained by their application. The psychometric instruments are developed in order to assess psychological constructs that cannot be operationally defined and, consequently, cannot be objectively assessed. According to \cite{law1998}, a construct is said to be multidimensional when it consists of a number of interrelated attributes or dimensions and exists in multidimensional domains. In order to develop a psychometric instrument to assess a multidimensional construct, a set of items, that assess a dimension, is developed for each one of its dimensions in furtherance of assessing the construct as a whole. The validation process of an instrument must guarantee that each item assesses its dimension correctly according to the desirable characteristics of a psychometric instrument, e.g., reliability and trustworthiness \cite{haynes1995}.

The validity of an instrument is divided in four categories: predictive validity, concurrent validity, content validity and construct validity. The first two of these may be together considered as criterion-oriented validation processes \cite{cronbach1955}. The predictive validity is studied when the instrument aims to predict a criterion. The instrument is applied and a correlated construct to the criterion is assessed, providing a prediction for the criterion of interest. The concurrent validity is studied when the instrument is proposed as a substitute for another \cite{cronbach1955}. The study of the construct validity of a psychometric instrument is necessary when the result of the instrument is the measure of an attribute or a characteristic that is not operationally defined. According to the construct validity, an instrument is valid when it is possible to determine which construct accounts for the variance of the instrument performance. 

Content validity is established by showing that the instrument items are a sample of a universe in which the investigator is interested. Content validity is ordinarily to be established deductively, by defining a universe of items and sampling systematically within this universe to establish the instrument \cite{cronbach1955}. Another definition for content validity is that it is the degree to which elements of an assessment instrument are relevant to and representative of the targeted construct for a particular assessment purpose \cite{haynes1995}. 

A list consisting of thirty-five procedures for the content validation was proposed by \cite{haynes1995}. Amidst these procedures are to match each item to the dimension of the construct that it assesses and request the judgement of specialists in the construct, also called judges, about the developed items. The accomplishment of these procedures is imperative to verify if the developed items are a sample of the universe that the instrument aims to assess. These procedures, components of the theoretical analysis of items, are subjective for they rely upon the personal opinions of specialists and researchers. Indeed, the theoretical analysis of items is done by judges and aims to establish the comprehension of the items (semantic analysis) and their pertinence to the attribute that they propose to assess.

This paper aims to propose a nonparametric statistical approach to the content analysis of items in furtherance of assessing its consistency and reliability. Therefore, our approach does not seek to establish the validity of the instrument, but rather assess the consistency of the content analysis process, so that its rule about the instrument may be trusted. Thus, this approach must be applied among other instrument validation methods, quantitative and qualitative, e.g., semantic analysis, pretrial and factorial analysis, in order to ensure the reliability, consistency, validity and trustworthiness of the psychometric instrument.

\section{Method}

The researcher, supported by the theory of the construct that the instrument aims to assess, develops \textit{m} items and for each item assigns a theoretical dimension according to the theory and/or his opinion about which dimension the item assesses. Although the items and their dimensions have theoretical foundations, it is necessary to test them in order to determine if every item is indeed assessing the dimension it is supposed to.

In order to fulfil such test, the items are sent to \textit{s} specialists in the construct, so that they may judge the items according to the dimension they assess. The items may be sent to at least six specialists and should be presented to them in a random order and without their theoretical dimensions, so that their judgement is not biased.

A condition for an item to be excluded from the instrument is determined based on the judgement of the specialists. This condition must exclude the items that do not belong to the universe that the instrument aims to assess, so that the not excluded items are a sample of such universe. A possible way to proceed is to determine a \textit{Concordance Index} (\textit{CI}) that states that all items in which less than $c$\% of the specialists agree on the dimension that they assess must be excluded. One may also take the \textit{Content Validity Ratio} (\textit{CRV}), as proposed by \cite{lawshe1975}, as a condition to exclude items that do not belong to the universe that the instrument aims to assess. 

The method to be developed in this paper aims to determine if all specialists have the same capability to judge the items according to their dimensions, through the analysis of the judgement of the specialists about the items that were not excluded by the established condition. However, the method does not rank the specialists according to their capabilities, but only determine if all specialists have the same capability. Therefore, it is not possible to determine the specialists with low capability. 

If there are no evidences that the capabilities of the specialists are different, their judgement is accepted and the items not excluded by the established condition are used in the next steps of the instrument validation process. Indeed, if all the specialists have the same capability, it may happen that they are all highly capable or little capable  of judging the items, though the proposed method will not be able to differentiate between the two cases. Nevertheless, the two scenarios may be differentiated by a qualitative analysis of the specialists judgements, by observing if they agree with the theoretical dimension of the items and, when they do not agree, if there is some theory that supports their choice. Therefore, if their judgements are consistent with some theory, then the specialists may be regarded as being all highly capable of judging the items, given that they have all the same capability to judge them.

On the other hand, if it is determined that the specialists do not have all the same capability to judge the items, then at least one specialist is less capable to judge them than the others, what may bias the validation of the instrument. Therefore, in such scenario, we propose two approaches in order to avoid a biased validation process. First, we propose that the specialists judgement be disregarded and a new group of specialists be requested to judge the items. However, this approach may be impractical in some cases, as time and resources may be too limited to repeat the cycle of specialists judgements more than once. Nonetheless, we propose a much more practical approach that consists in applying the proposed method to all subgroups of specialists of size $s^{*}$, $6 \leq s^{*} < s$, of the original group of specialists, and then choose the judgement of the subgroup whose specialists have all the same capability to judge the items. This approach will be presented in more details in the application section.

\section{Notation and Definitions}

Let $C = \{C_{1},\dots, C_{n}\}$ be a construct divided in \textit{n} dimensions and \textit{U} be the universe of all the items that assess the dimensions of \textit{C}. A set $I = \{i_{1}, \dots, i_{m}\}$ of items is developed based on the theory about \textit{C} and then a subset $I^{*} \subset I$ of items, that we believe to be a subset of \textit{U}, is determined, by the following process.

Denote $E = \{e_{1},\dots,e_{s}\}$ a set of $s$ specialists and let $C_{c(i_{l})} \in C$ be the dimension that the item $i_{l} \in I^{*}$ assesses. Let the random variables $\{X_{i_{l}}(e_{j}):i_{l} \in I,e_{j} \in E\}$, defined on $(\Omega, \mathbb{F},\mathbb{P})$, be so that $X_{i_{l}}(e_{j}) = k$ if the specialist $e_{j}$ judged the item $i_{l}$ at the \textit{kth} dimension of \textit{C}. Note that if $i_{l} \in I^{*}$ and $X_{i_{l}}(e_{j}) = c(i_{l})$, then the specialist $e_{j}$ judged the item $i_{l}$ correctly.

The capability of the specialist $e_{j}$ to judge the items  is defined as
\begin{equation*}
\boldsymbol{P}(e_{j}) = (P_{i_{l}}(e_{j}): i_{l} \in I^{*})
\end{equation*} 
in which $P_{i_{l}}(e_{j}) = \mathbb{P}\{X_{i_{l}}(e_{j}) = c(i_{l})\}, \forall i_{l} \in I^{*}$ and $\forall e_{j} \in E$. In the proposed approach, we are interested in developing a hypothesis test to determine if $\boldsymbol{P}(e_{j}) = \boldsymbol{p} \in [0,1]^{|I^{*}|}, \forall e_{j} \in E$, i.e., if all specialists have the same capability to judge the items.

For this purpose, let a random sample of the judgement of the specialist $e_{j}$ about the items of $I$ be given by $\boldsymbol{x}_{e_{j}} = \{x_{i_{1}}(e_{j}), \dots, x_{i_{m}}(e_{j})\}$ and let $\boldsymbol{X}$ be the space of all possible random samples $\{\boldsymbol{x}_{e_j}: e_{j} \in E\}$. Define the random sets $\{M_{i_{l}}: i_{l} \in I\}$ as
\begin{equation*}
M_{i_{l}} = \arg\max\limits_{k \in \{1, \dots, n\}} \Bigg\{\sum_{e_{j} \in E} \mathds{1}_{\{k\}} \Big(X_{i_{l}}(e_{j})\Big) \Bigg\}
\end{equation*}
in which $\mathds{1}_{\{A\}}(\cdot)$ is the indicator function of the set $A$. Note that $M_{i_{l}}$ is the set containing the number of the dimensions in which the majority of the specialists judged the item $i_{l} \in I$. Given a random sample $\{\boldsymbol{x}_{e_j}: e_{j} \in E\} \in \boldsymbol{X}$ and a subset $I^{*} \subset I$ of items, the set $\{m_{i_{l}}: i_{l} \in I^{*}\}$, determined from the sample values $\{\boldsymbol{x}_{e_j}: e_{j} \in E\}$, is a random sample of $\{M_{i_{l}}: i_{l} \in I^{*}\}$.

The subset $I^{*}$ may be defined by a condition function, a function of the sample $\{\boldsymbol{x}_{e_j}: e_{j} \in E\}$, given by $f: \boldsymbol{X} \mapsto \mathcal{P}(I)$, in which $\mathcal{P}(\cdot)$ is the power set operator. The condition function must be so that if $\{m_{i_{l}}: i_{l} \in I^{*}\}$ is determined from $\{\boldsymbol{x}_{e_j}: e_{j} \in E\} \in \boldsymbol{X}$ and $I^{*} = f(\boldsymbol{x}_{e_j}: e_{j} \in E)$, then $|m_{i_{l}}| = 1, \forall i_{l} \in I^{*}$. The \textit{CI} for $c > 50$ and the $CRV$ are condition functions. From now on, it will be supposed that the condition function may be expressed as a \textit{CI}.

The condition function is based on the assumption that an item is in the universe of items that assess the construct of interest if the majority of specialists agree on the dimension it assesses. Of course, one may take a different criterion to exclude the items that do not assess the construct of interest, although our method may be applied only if the criterion can be expressed as a condition function, for it is based on the fact that $M_{i_{l}}$ is a univariate random variable.

Finally, define
\begin{equation*}
W_{i_{l}}(e_{j}) = \mathds{1}_{\{M_{i_{l}}\}} \Big(X_{i_{l}}(e_{j})\Big),
\end{equation*}
as the random variable that indicates if the specialist $e_{j}$ judged the item $i_{l}$ at the same dimension as the majority of the specialists. Given a random sample $\{\boldsymbol{x}_{e_j}: e_{j} \in E\} \in \boldsymbol{X}$ and a subset $f(\boldsymbol{x}_{e_j}: e_{j} \in E) = I^{*} \subset I$ of items, the set $\{w_{i_{l}}(e_{j}): i_{l} \in I^{*}, e_{j} \in E\}$, determined from the sample values $\{\boldsymbol{x}_{e_j}: e_{j} \in E\}$, is a random sample of $\{W_{i_{l}}(e_{j}): i_{l} \in I^{*}, e_{j} \in E\}$.

On the one hand, whilst we observe the values of the random variables $\{X_{i_{l}}(e_{j}):i_{l} \in I,e_{j} \in E\}$, we do not know if the specialists judged the items correctly or not, for the dimension that an item really assesses (if any) is unknown. Therefore, it is not possible to differentiate the specialists by the number of items they judged correctly, for example.

On the other hand, from the random variables $\{W_{i_{l}}(e_{j}): i_{l} \in I, e_{j} \in E\}$, we know the concordance of the specialists on the judgement of the items, what gives us a relative measure of the capability of the specialists to judge the items. Therefore, we are able to test if all the specialists have the same capability to judge the items, although we cannot determine the capability of each one.

\section{Assumptions} 

The development of the items and the judgement of the specialists must satisfy two assumptions so that the method to be presented below may be applied:
\begin{enumerate}
	\item[\textbf{1.}] Each item $i_{l} \in I^{*}$ assesses one, and only one, dimension $C_{c(i_{l})} \in C$.
	\item[\textbf{2.}] The random variables $\{X_{i_{l}}(e_{j}): i_{l} \in I^{*}, e_{j} \in E\}$ are independent.
\end{enumerate}

Assumption \textbf{1} establishes that the items that were not excluded by the condition function, i.e., the items in $I^{*}$, are well constructed and assess only one dimension of \textit{C}, while assumption \textbf{2} imposes that the specialists judge the items independently of each other and that the judgement of a specialist about one item does not depend on his judgement about any other item. Those assumptions are not strong, for it is expected that they will be satisfied if the items were well constructed. Indeed, as better the condition function is in determining what items are not in \textit{U}, the better will be  the quality of the items in $I^{*}$. Therefore the assumptions above are closely related to the condition function. If, in fact, $I^{*} \subset U$, then the first assumption is immediately satisfied, for there is no intersection between two dimensions of a construct, and the second assumption may also hold, for the items are well defined.

\section{Mathematical Deduction}

Given a random sample $\{\boldsymbol{x}_{e_j}: e_{j} \in E\} \in \boldsymbol{X}$, it is not trivial to estimate the capabilities $\{\boldsymbol{P}(e_{j}): e_{j} \in E \}$, for the dimension that each item assesses is unknown. Examining such random sample, it is known that the specialist $e_{j}$ judged the item $i_{l}$ at the dimension $C_{k}$, but it is not possible to determine, with probability 1, if he judged such item correctly. Therefore, the problem is, given a random sample $\{\boldsymbol{x}_{e_j}: e_{j} \in E\} \in \boldsymbol{X}$, to determine random variables that allow us to test if the capability of all the specialists is the same. It will be shown that if the random variables $\{W_{i_{l}}(e_{j}): e_{j} \in E \}$ are not identically distributed $\forall i_{l} \in I^{*}$, then the specialists do not have all the same capability to judge the items. Indeed, in order to test if the capability of all specialists is the same, we will consider the following null hypotheses:
\begin{equation*}
H_{0}: \begin{cases}
\text{\textbf{1.} } & \boldsymbol{P}(e_{j}) = (p^{(i_{1})}, \dots, p^{(i_{|I^{*}|})}) = \boldsymbol{p} \in [0,1]^{I^{*}}, \forall e_{j} \in E\\
\text{\textbf{2.} } & \Big(\mathbb{P}\{X_{i_{l}}(e_{j}) = 1\}, \dots, \mathbb{P}\{X_{i_{l}}(e_{j}) = n\}\Big) \text{ is a permutation of } \\ & \Big(p^{(i_{l})},p_{1}^{(i_{l})}, \dots, p_{n-1}^{(i_{l})}\Big) \in [0,1]^{n}, p^{(i_{l})} + \sum_{k=1}^{n-1} p_{k}^{(i_{l})} = 1, \forall i_{l} \in I^{*}, \forall e_{j} \in E \\
\end{cases}
\end{equation*}
Of course, we are only interested in testing the first part of $H_{0}$, that refers to the capability of the specialists, i.e., that all specialists have the same capability to judge the items. However, the second part is needed to develop a test statistic for $H_{0}$. It will be argued that for great values of $p^{(i_{l})}$ the hypothesis that is actually being tested is the first one.

The propositions below set the scenario for the nonparametric test that will be used to test $H_{0}$.

\begin{proposition}
	\label{P1}
	The random variables $\{W_{i_{l}}(e_{j}): i_{l} \in I^{*}\}$ are independent $\forall e_{j} \in E$, but the random variables $\{W_{i_{l}}(e_{j}): e_{j} \in E\}$ are dependent $\forall i_{l} \in I^{*}$.
\end{proposition}

\begin{proof}
	On the one hand, the random variables $\{W_{i_{l}}(e_{j}): i_{l} \in I^{*}\}$ are each, by assumption \textbf{2}, function of independent random variables, therefore they are independent. On the other hand, note that $\sum\limits_{e_{j} \in E} W_{i_{l}}(e_{j}) \geq \lceil \frac{cs}{100} \rceil$, for at least $c\%$ of the specialists must agree on the dimension an item in $I^{*}$ assesses, what establishes a dependence.
\end{proof}

\begin{proposition}
	\label{P2}
	Under $H_{0}$, the random variables $\{W_{i_{l}}(e_{j}): e_{j} \in E\}$ are identically distributed for all $i_{l} \in I^{*}$.
\end{proposition}

\begin{proof}
	We have that
	\begin{align*}
	\mathbb{P}\{W_{i_{l}}(e_{j}) = 1\} & = \mathbb{P}\{X_{i_{l}}(e_{j}) = M_{i_{l}}\} \\ & = \mathbb{P}\{X_{i_{l}}(e_{j}) = c(i_{l}), M_{i_{l}} = c(i_{l})\} + \mathbb{P}\{X_{i_{l}}(e_{j}) = M_{i_{l}}, M_{i_{l}} \neq c(i_{l})\}. 
	\end{align*}
	Now let $X^{(i_{l})} \sim Binomial(s-1,p^{(i_{l})})$ and $X_{k}^{(i_{l})} \sim Binomial(s-1,p_{k}^{(i_{l})}), k \in \{1, \dots, n-1\}$, be independent random variables, and let $f^{*} = \lfloor \frac{cs}{100} \rfloor$, in which \textit{c} is the \textit{CI}. Then,
	\begin{align*}
	\mathbb{P}\{X_{i_{l}}(e_{j}) = c(i_{l}), M_{i_{l}} = c(i_{l})\} & =  \mathbb{P}\{ M_{i_{l}} = c(i_{l})|X_{i_{l}}(e_{j}) = c(i_{l})\}\mathbb{P}\{X_{i_{l}}(e_{j}) = c(i_{l})\}\\  & = \mathbb{P}\{X^{(i_{l})} \geq f^{*}\} p^{(i_{l})}
	\end{align*}
	and
	\begin{align*}
	\mathbb{P}\{X_{i_{l}}(e_{j}) = M_{i_{l}}, M_{i_{l}} \neq c(i_{l})\} & = \sum\limits_{\substack{k=1 \\ k \neq c(i_{l})}}^{n} \mathbb{P}\{X_{i_{l}}(e_{j}) = k, M_{i_{l}} = k\} \\ & = \sum\limits_{\substack{k=1 \\ k \neq c(i_{l})}}^{n} \mathbb{P}\{ M_{i_{l}} = k|X_{i_{l}}(e_{j}) = k\}\mathbb{P}\{X_{i_{l}}(e_{j}) = k\} \\ & = \sum\limits_{k=1}^{n-1} \mathbb{P}\{X_{k}^{(i_{l})} \geq f^{*}\} p_{k}^{(i_{l})}.
	\end{align*}
	Hence,
	\begin{align*}
	\mathbb{P}\{W_{i_{l}}(e_{j}) = 1\} & =  \mathbb{P}\{X^{(i_{l})} \geq f^{*}\} p^{(i_{l})} + \sum\limits_{k=1}^{n-1} \mathbb{P}\{X_{k}^{(i_{l})} \geq f^{*}\} p_{k}^{(i_{l})}
	\end{align*}
	that does not depend on $e_{j}$ and the result follows.
\end{proof}

It is important to note that if all $p^{(i_{l})}$ are approximately $1$, then $\mathbb{P}\{W_{i_{l}}(e_{j}) = 1\} \approx p^{(i_{l})} \mathbb{P}\{X \geq f^{*}\}$ and the hypothesis that is really being tested is the first part of $H_{0}$. Therefore, it is reasonable to test $H_{0}$ in order to determine if the specialists have all the same capability to judge the items, for, if it is indeed true, we expect that all $p^{(i_{l})}$ are great and the second part of $H_{0}$ will hardly leads to the rejection of $H_{0}$ when the capability is the same.

This test may be used as a diagnostic for the content analysis of items. If $H_{0}$ is not rejected, then there is no evidence that the capabilities of the specialists are different. However, if $H_{0}$ is rejected, we do not know if it is the first or the second part (or both) of $H_{0}$ that is not being satisfied by the judgement of the specialists. Nevertheless, we may disregard the judgement of those specialists in any case, for either their capability is not the same or they are the same, but some $p^{(i_{l})}$ are small, what led to the rejection of $H_{0}$ by its second part.

\section{Hypothesis Testing}

The Chocran's Q test may applied to the random sample $\{w_{i_{l}}(e_{j}): i_{l} \in I^{*}, e_{j} \in E\}$ determined from $\{\boldsymbol{x}_{e_j}: e_{j} \in E\}$ as a way to test $H_{0}$ \cite{cochran1950}. The assumptions of the Chocran's Q test, using the notation of this paper, are:
\begin{itemize}
	\item[\textbf{(a)}] The items of $I^{*}$ were randomly selected from the items that form the universe \textit{U} that the instrument aims to assess.
	\item[\textbf{(b)}] The random variables $\{W_{i_{l}}(e_{j}): i_{l} \in I^{*}, e_{j} \in E\}$ are dichotomous.
	\item[\textbf{(c)}] The random variables $\{W_{i_{l}}(e_{j}): i_{l} \in I^{*}\}$ are independent.
\end{itemize}

From the usual scenario in which such test is applied, we have that the items may be seen as the \textit{blocks} and the specialists as the \textit{treatments}. What the Chocran's Q test evaluates is if the random variables $\{W_{i_{l}}(e_{j}): e_{j} \in E\}$ are identically distributed for all $i_{l} \in I^{*}$. Therefore, if we reject the null hypothesis of the test, we conclude that $\{W_{i_{l}}(e_{j}): e_{j} \in E\}$ are not identically distributed for all $i_{l} \in I^{*}$ and, by Proposition \ref{P2}, $H_{0}$ is also rejected. Thus, the hypothesis tested by the Chocran's Q test is indeed $H_{0}$.

The statistic of the test is calculated from Table \ref{T1}, in which $I^{*} = \{i^{*}_{1}, \dots,i^{*}_{v}\}$, and may be expressed as
\begin{equation*}
Q = \sum\limits_{r =1}^{s} \dfrac{s(s-1)\Big(D_{r} - \frac{N}{s}\Big)^{2}}{\sum_{l=1}^{v} R_{l}(s-R_{l})}
\end{equation*}

\begin{table}[H]
	\centering
	\caption{Table of the observed random sample.}
	\label{T1}
	\begin{tabular}{c|ccc|c}
		\hline
		\multirow{2}{*}{Item} & \multicolumn{3}{c|}{Specialist} & \multirow{2}{*}{Total} \\ \cline{2-4}
		& $e_{1}$ & $\cdots$ & $e_{s}$ & \\
		\hline
		$i^{*}_{1}$ & $w_{i^{*}_{1}}(e_{1})$ & $\cdots$ & $w_{i^{*}_{1}}(e_{s})$ & $R_{1} = \sum\limits_{e_{j} \in E} w_{i^{*}_{1}}(e_{j})$ \\
		$\vdots$ & $\vdots$ & $\vdots$ & $\vdots$ & $\vdots$ \\
		$i^{*}_{v}$ & $w_{i^{*}_{v}}(e_{1})$ & $\cdots$ & $w_{i^{*}_{v}}(e_{s})$ & $R_{v} = \sum\limits_{e_{j} \in E} w_{i^{*}_{v}}(e_{j})$ \\
		\hline 
		Total & $D_{1} = \sum\limits_{i_{l} \in I^{*}} w_{i_{l}}(e_{1})$ & $\cdots$ & $D_{s} = \sum\limits_{i_{l} \in I^{*}} w_{i_{l}}(e_{s})$ & $N = \sum\limits_{i_{l} \in I^{*}} \sum\limits_{e_{j} \in E} w_{i_{l}}(e_{j})$ \\
		\hline
	\end{tabular}
\end{table}

The exact distribution of the $Q$ statistics may be calculated by the method presented by \cite{patil1975}, although a large sample approximation may be used instead. If $|I^{*}|$ is large, then the distribution of $Q$ is approximately $\chi^{2}$ with $(s-1)$ degrees of freedom \cite{conover1998}.

It is worth mentioning that the random variables $\{W_{i_{l}}(e_{j}): e_{j} \in E\}$ being identically distributed for all $i_{l} \in I^{*}$ does not imply that the specialists have all the same capability to judge the items, although there is no evidence that their capabilities are different. If there is no evidence that the capabilities of the specialists to judge the items are different, their judgement may be accepted.

If it is determined that the random variables $\{W_{i_{l}}(e_{j}): e_{j} \in E\}$ are not identically distributed for all $i_{l} \in I^{*}$, then the judgement of the specialists is disregarded for $H_{0}$ is rejected. The items may be judged by different groups of specialists until they are judged by a set in which all the specialists have the same capability to judge the items. Those groups may be formed by new specialists or may be a subgroup of size $s^{*}$, $6 \leq s^{*} < s$, of the specialists for which $H_{0}$ was rejected.

\section{Simulation Study}

As the Cochran's Q test is not a powerful one, i.e., its Type I error may be too great, a simulation study will be conducted to estimate the power of the test in some specific cases. The power of a statistical test is defined as the probability of $H_{0}$ being rejected when it is false and depends on the real scenario, i.e., on the real values of the parameters considered on $H_{0}$. Therefore, the power of Cochran's Q test in testing $H_{0}$ depends on the real capability of each specialist in judging the items, so that the simulation study consider \textit{10} distinct scenarios and is conducted as follows. 

For each scenario, we will simulate \textit{50,000} judgements of the same items by the judges and then determine the proportion of the simulations in which $H_{0}$ was rejected at a significance, i.e., Type II error, of 5\%. This proportion will be regarded as an estimate for the power of the test in the considered scenario. A CI of \textit{50\%} will be used to determine $I^{*}$ in each simulation. Analysing the results of all \textit{10} scenarios, we will have a wide picture of the power of the test and will know for which scenarios it is more powerful.

We will consider in all scenarios nine specialists judging \textit{30} items into three dimensions, that is the framework of the application in the next section. We will also consider that the capability of each specialist is the same for all items, i.e., that $\mathbb{P}\{X_{i_{l}}(e_{j}) = c(i_{l})\} = p_{j}$ for all $j \in \{1,\dots,9\}$ and $l \in \{1,\dots,|I^{*}|\}$. Finally, we will assume that $\mathbb{P}\{X_{i_{l}}(e_{j}) = k\} = (1 -p_{j})/2$ for all $k \neq c(i_{l})$, $j \in \{1,\dots,9\}$ and $l \in \{1,\dots,|I^{*}|\}$. The scenarios and their estimated test power are displayed in Table \ref{simulations}.

\begin{table}[ht]
	\centering
	\caption{The estimated power of the test for each scenario.}
	\label{simulations}
	\begin{tabular}{c|lcc}
		\hline
		Scenario & Description & Items$^{*}$ & Power \\
		\hline
		1 & $p_{j} = 0.9$, $j \neq 1$ and $p_{1} = 0.45$ & 30 & 0.9931 \\
		2 & $p_{j} = 0.9$, $j \notin \{1,2,3\}$ and $p_{1} = p_{2} = p_{3} = 0.45$ & 30 & 0.9999 \\
		3 & $p_{j} = 0.9$, $j \notin \{1,2,3\}$ and $p_{1} = 0.45, p_{2} = 0.35, p_{3} = 0.25$ & 29 & 1 \\
		4 & $p_{j} = 0.9$, $j \neq 1$ and $p_{1} = 0.8$ & 30 & 0.1595 \\
		5 & $p_{j} = 0.9$, $j \notin \{1,2\}$ and $p_{1} = p_{2} = 0.8$ & 30 & 0.2413 \\
		6 & $p_{j} = 0.6$, $j \notin \{1,2,3\}$ and $p_{1} = p_{2} = p_{3} = 0.75$ & 29 & 0.2413 \\
		7 & $p_{j} = 0.3$, $j \notin \{1,2,3\}$ and $p_{1} = p_{2} = p_{3} = 0.75$ & 13 &  0.3167 \\
		8 & $(p_{1},\dots,p_{9}) = (0.1,0.2,0.3,0.4,0.5,0.6,0.7,0.8,0.9)$ & 15 & 0.4571 \\
		9 & $p_{j} = 0.9, \forall j$, but the second part of $H_{0}$ is not true & 30 & 0.0456 \\
		10 & $p_{j} = 0.6, \forall j$, but the second part of $H_{0}$ is not true & 26 & 0.0553 \\
		\hline
		\multicolumn{4}{l}{$^{*}$ The mean number of items not excluded by the \textit{CI}.}
	\end{tabular}
\end{table}

On the one hand, we see in Table \ref{simulations}, that the power of the test is great when the majority of the specialists have the same high capability, while a few specialists have a low capability, as is the case of scenarios 1, 2 and 3. On the other hand, the power of the test is quite low when some of the specialists have the same high capability, but the specialists with lower capability are almost as capable as them, as is the case of scenarios 4, 5 and 6.

On scenarios 7 and 8 we see that the power of the test is low when there are specialists with capability less than \textit{0.5}. It happens because the specialists hardly agree on the dimension that each item assesses (as some of them are not capable) so that many items are excluded by the \textit{CI} and, on the items that remain, the not capable specialists agree with the highly capable ones, so it seems that they have high capability. Indeed, in scenarios 7 and 8, the mean number of not excluded items are the lowest of all scenarios, so that a low concordance among the specialists is an evidence of the existence of low capable specialists, given that the items were well constructed.

Finally, as pointed out in the Mathematical Deduction section, we see in scenarios 9 and 10 that the hypotheses that is actually being tested when all the specialists are highly and equally capable is the first part of $H_{0}$, as the power of the test is close to the Type II error, what must be the case if the hypothesis is true. 

The simulation study shed light in some interesting facts about the proposed method on the considered scenarios. On the one hand, if the majority of the specialists have a homogeneous high capability, and a few specialists have a very low capability, then the power of the test is great. However, if the specialists have all high, but different, capability then the power of the test is low. On the other hand, if the majority of the specialists have a low capability, then a great number of items will be excluded by the \textit{CI} and, given that the items were well constructed, we may conclude that the specialists have low capability of judging the items, even though the power of the test is low. Finally, if only the first part of $H_{0}$ is being satisfied, and the capability of the specialists is high, then the power of the test is low and, therefore, the hypothesis that is really being tested is the first part of $H_{0}$.

\section{Application: Perception About the Evaluation of the Teaching-Learning}

In this section we will apply the developed method to a real validation process, in order to analyse the content of items of an instrument that aims to assess the perception of teachers and students of higher education institutions about the teaching-learning process, that is a construct that may be divided into three dimensions: process (P), judgement (J) and teaching-learning (T). 

The evaluation of the teaching-learning has a process dimension, as it must have a beginning, a middle and an end well defined and must have a continuous, cumulative and systematic character. Indeed, it is a systematic mechanism for gathering information over time, with well defined levels, what characterizes it as a process. Also, the evaluation of the teaching-learning has a judgement dimension because it must issue a judgement of value or assign a score through the analysis of educational results obtained from the information gathered over the time. Finally, the evaluation of the teaching-learning has a teaching-learning dimension for, as its own name says, it must not only evaluate the learning, but also the teaching: it should not only evaluate what the student has learnt, but also what the teacher has taught. Therefore, the evaluation of the teaching-learning is a process of data gathering, in which an individual will judge or be judged accordingly to the teaching-learning.

In order to develop an instrument to assess this construct, \textit{30} items were developed and sent to nine specialists so they would judge the items to the dimension that, according to their opinion, each one assesses. The condition defined for excluding an item is the \textit{CI} with $c = 50$. The judgements of the specialists are presented in Table \ref{judgement}, the table for the Cochran's Q test is displayed in Table \ref{Q} and a translation of the items, that were originally constructed in Portuguese, is presented in the Appendix. 

\begin{table}[ht]
	\centering
	\caption{Judgement of the specialists about each item, i.e., the sample $\{\boldsymbol{x}_{e_j}: e_{j} \in E\}$.}
	\label{judgement}
	\begin{tabular}{c|ccccccccc|cc}
		\hline
		\multirow{2}{*}{Item}& \multicolumn{9}{c|}{Specialist} & \multirow{2}{*}{Dimension$^{*}$} & \multirow{2}{*}{Theoretical} \\ \cline{2-10}
		& 1 & 2 & 3 & 4 & 5 & 6 & 7 & 8 & 9 & & \\ 
		\hline
		1 & T & P & P & T & T & T & T & P & P & T & P \\ 
		2 & P & J & T & T & P & P & T & P & T & - & T \\ 
		3 & T & T & J & P & P & P & P & J & J & - & J \\ 
		4 & J & P & P & P & P & P & P & J & J & P & P \\ 
		5 & T & J & T & T & T & T & J & T & P & T & P \\ 
		6 & P & T & T & P & T & T & T & P & T & T & P \\ 
		7 & J & J & J & J & J & J & J & J & J & J & J \\ 
		8 & T & P & P & T & P & P & T & P & T & P & P \\ 
		9 & J & P & T & J & T & T & P & P & P & - & T \\ 
		10 & J & J & J & J & J & J & J & J & J & J & J \\ 
		11 & J & T & J & J & J & J & J & J & J & J & J \\ 
		12 & P & T & T & P & P & P & P & J & P & P & P \\ 
		13 & J & P & P & T & J & J & J & J & T & J & J \\ 
		14 & T & T & T & T & P & P & T & P & T & T & P \\ 
		15 & J & P & J & J & J & J & J & J & J & J & J \\ 
		16 & T & P & T & T & P & P & J & P & T & - & P \\ 
		17 & P & P & P & P & T & T & P & T & P & P & T \\ 
		18 & J & T & T & T & P & P & J & T & J & - & T \\ 
		19 & T & T & T & T & P & P & T & J & P & T & T \\ 
		20 & P & P & P & T & P & P & P & J & J & P & P \\ 
		21 & J & J & J & J & J & J & J & J & P & J & P \\ 
		22 & P & P & P & P & P & P & T & T & P & P & P \\ 
		23 & J & J & J & J & J & J & J & J & J & J & J \\ 
		24 & T & J & P & T & P & P & T & J & T & - & J \\ 
		25 & J & J & J & J & J & J & J & J & J & J & J \\ 
		26 & T & P & T & T & P & P & P & P & T & P & T \\ 
		27 & P & P & P & T & P & P & T & J & P & P & T \\ 
		28 & T & P & P & J & P & P & J & P & T & P & T \\ 
		29 & T & T & P & T & P & P & P & P & T & P & T \\ 
		30 & T & J & J & J & T & T & J & J & P & J & J \\  
		\hline
		\multicolumn{12}{l}{$^{*}$ The dimension on which at least 50\% of the specialists agree that} \\
		\multicolumn{12}{l}{the item assesses.}
	\end{tabular}
\end{table}

\begin{table}[ht]
	\centering
	\caption{The table for the Cochran's Q test, i.e., the sample $\{w_{i_{l}}(e_{j}): i_{l} \in I^{*}, e_{j} \in E\}$.}
	\label{Q}
	\begin{tabular}{c|ccccccccc|c}
		\hline
		\multirow{2}{*}{Item}& \multicolumn{9}{c|}{Specialist} & \multirow{2}{*}{Total} \\ \cline{2-10}
		& 1 & 2 & 3 & 4 & 5 & 6 & 7 & 8 & 9 & \\ 
		\hline
		1 & 1 & 0 & 0 & 1 & 1 & 1 & 1 & 0 & 0 & 5 \\ 
		2 & 1 & 0 & 0 & 0 & 1 & 1 & 0 & 1 & 0 & 4 \\ 
		3 & 0 & 0 & 0 & 1 & 1 & 1 & 1 & 0 & 0 & 4 \\ 
		4 & 0 & 1 & 1 & 1 & 1 & 1 & 1 & 0 & 0 & 6 \\ 
		5 & 1 & 0 & 1 & 1 & 1 & 1 & 0 & 1 & 0 & 6 \\ 
		6 & 0 & 1 & 1 & 0 & 1 & 1 & 1 & 0 & 1 & 6 \\ 
		7 & 1 & 1 & 1 & 1 & 1 & 1 & 1 & 1 & 1 & 9 \\ 
		8 & 0 & 1 & 1 & 0 & 1 & 1 & 0 & 1 & 0 & 5 \\ 
		9 & 0 & 1 & 0 & 0 & 0 & 0 & 1 & 1 & 1 & 4 \\ 
		10 & 1 & 1 & 1 & 1 & 1 & 1 & 1 & 1 & 1 & 9 \\ 
		11 & 1 & 0 & 1 & 1 & 1 & 1 & 1 & 1 & 1 & 8 \\ 
		12 & 1 & 0 & 0 & 1 & 1 & 1 & 1 & 0 & 1 & 6 \\ 
		13 & 1 & 0 & 0 & 0 & 1 & 1 & 1 & 1 & 0 & 5 \\ 
		14 & 1 & 1 & 1 & 1 & 0 & 0 & 1 & 0 & 1 & 6 \\ 
		15 & 1 & 0 & 1 & 1 & 1 & 1 & 1 & 1 & 1 & 8 \\ 
		16 & 1 & 0 & 1 & 1 & 0 & 0 & 0 & 0 & 1 & 4 \\ 
		17 & 1 & 1 & 1 & 1 & 0 & 0 & 1 & 0 & 1 & 6 \\ 
		18 & 0 & 1 & 1 & 1 & 0 & 0 & 0 & 1 & 0 & 4 \\ 
		19 & 1 & 1 & 1 & 1 & 0 & 0 & 1 & 0 & 0 & 5 \\ 
		20 & 1 & 1 & 1 & 0 & 1 & 1 & 1 & 0 & 0 & 6 \\ 
		21 & 1 & 1 & 1 & 1 & 1 & 1 & 1 & 1 & 0 & 8 \\ 
		22 & 1 & 1 & 1 & 1 & 1 & 1 & 0 & 0 & 1 & 7 \\ 
		23 & 1 & 1 & 1 & 1 & 1 & 1 & 1 & 1 & 1 & 9 \\ 
		24 & 1 & 0 & 0 & 1 & 0 & 0 & 1 & 0 & 1 & 4 \\ 
		\hline
		Total & 18 & 14 & 17 & 18 & 17 & 17 & 18 & 12 & 13 & 144 \\ 
		\hline
	\end{tabular}
\end{table}

The statistic of the Cochran's Q test for the data in Table \ref{Q} is $Q = 8.7$ and the test p-value is $0.36$, so that there is no evidence that $H_{0}$ is not true, at a significance of 5\%. Furthermore, as the majority of the specialists agreed on the dimension that \textit{24} out of \textit{30} (80\%) items assess we also do not have evidence that the capability of the specialists is low. Therefore, based on the proposed method, there is no reason to disregard the judgement of the specialists.

Nevertheless, in order to illustrate the proposed approach for the case in which $H_{0}$ is rejected, we will apply the test to every subgroup of size $6 \leq s^{*} < 9$ of specialists, what amounts to \textit{130} subgroups, and see for what subgroups the capability of the specialists is the same. From the \textit{130} subgroups, for \textit{13} of them $H_{0}$ was rejected at a significance of 5\%. The $Q$ statistic and the p-value for the \textit{10} groups with greatest p-values are displayed in Table \ref{subgroups}. If $H_{0}$ had been rejected to the group of nine specialists we could now look for a subgroup of those specialists for which $H_{0}$ is not rejected and, by the help of a qualitative analysis, we could choose a subgroup of those specialists instead of disregarding their judgements as a whole and sending the items to other specialists to judge. 

\begin{table}[ht]
	\centering
	\caption{The result of the Cochran's Q test for the subgroups of specisliats with the biggest p-value.}
	\label{subgroups}
	\begin{tabular}{l|cc}
		\hline
		Specialists & Q & p-value \\ 
		\hline
		(2,3,5,6,7,8) & 0.778 & 0.978 \\ 
		(2,3,5,6,8,9) & 0.789 & 0.978 \\ 
		(1,2,3,4,5,6,8,9) & 2.478 & 0.929 \\ 
		(2,3,5,6,7,9) & 1.772 & 0.880 \\ 
		(2,3,4,5,7,8,9) & 2.441 & 0.875 \\ 
		(2,3,4,6,7,8,9) & 2.441 & 0.875 \\ 
		(1,3,4,5,6,7) & 1.923 & 0.860 \\ 
		(2,4,5,7,8,9) & 1.991 & 0.850 \\ 
		(2,4,6,7,8,9) & 1.991 & 0.850 \\ 
		(1,3,5,6,8,9) & 2.069 & 0.840 \\ 
		\hline
	\end{tabular}
\end{table}

\FloatBarrier
\section{Final Remarks}

The Cochran's Q test is not a powerful one, thus the method must be used with caution. The validation of a psychometric instrument is a process formed by various procedures, therefore it must not be restricted to the content analysis of items and the method developed in this paper. It is important to apply other validation techniques, both qualitative and quantitative, to the instrument so it may be properly validated.

The method may be improved in order to decrease even more the subjectivity of the content analysis of items, especially by the development of more powerful tests than the one established and the definition of other random variables that enable the comparison between the judgement of the specialists. This paper does not exhaust the subject, but present a nonparametric statistical approach that aims to decrease the subjectivity of a subjective process and that may applied not only to the content analysis of items, but also to any statistical application that enables the definition of variables such as those of this paper.

\section{Supplementary Materials}

The \textbf{R} \cite{R} code used in the simulation study and in the application section is available as supplementary material for this paper and can be accessed at \textit{www.ime.usp.br/$\sim$dmarcondes}.

\appendix
\section*{Appendix}
A translation of the constructed items, whose response is the Likert Scale, and their theoretical dimension, are presented below.

\begin{enumerate}
	\item The evaluation is an instrument strategically used to help on the difficulties (Process).
	\item The evaluation assumes a formative role on the teaching-learning (Teaching-Learning).
	\item The proposed evaluation methods are fair and appropriate (Judgement).
	\item The time available for evaluation is sufficient (Process).
	\item The evaluation offers recovery strategies for students that have difficulties (Process). 
	\item The instructions given for the assignments subjected to evaluation are useful (Process).  
	\item The evaluation is a tool for punishing the student in some manner (Judgement).
	\item The evaluation is an essential tool for the teaching-learning process (Process).
	\item The evaluation is an essential tool for the understanding of the taught subject (Teaching-Learning).
	\item The evaluation is a process that ranks the students is some manner (Judgement).
	\item The evaluation is a process that, in a particular way, builds a hierarchy among the students (Judgement).
	\item The evaluation is a process that follows the student during all his academic life (Process).
	\item The evaluation has different meanings for who evaluate and for who is evaluated (Judgement).
	\item The evaluation is used to find out where and how the teaching-learning may be improved (Process).
	\item The evaluation is a tool to reward the student in some manner (Judgement).
	\item The evaluation is a tool to diagnostic the teaching-learning (Process).
	\item The evaluation is a tool with technical and pedagogical characteristics (Teaching-Learning).
	\item The evaluation aims to identify how much the student has learnt the subjects (Teaching-Learning).
	\item The evaluation aims to identify which paths take to knowledge (Teaching-Learning).
	\item The evaluation is a systematic evidence gathering process (Process).
	\item The evaluation is a process of outlining, obtaining and providing informations that permit to judge decision alternatives (Process).
	\item The evaluation is a process with continuous, cumulative and systematic, but not episodic, character (Process).
	\item Evaluate means to provide a judgement of value or to assign a score to whom is being evaluated (Judgement).
	\item The evaluation is a tool that permits to inquire to what extent the defined objectives are being achieved (Judgement).
	\item The evaluation has an authoritarian and classificatory role inside the process of teaching-learning (Judgement).
	\item The evaluation is an educational component that can facilitate the teaching-learning (Teaching-Learning).
	\item The teaching-learning and the evaluation are not isolated parts of the education process (Teaching-Learning).
	\item The evaluation is the more adequate path to make it feasible an excellent teaching-learning (Teaching-Learning).
	\item The evaluation stimulate the acts of teaching and learning as a simultaneous process (Teaching-Learning).
	\item The evaluation involves the intentional judgement of a process developed by an individual, during your learning (Judgement).
\end{enumerate}

\bibliographystyle{amsplain}
\bibliography{mybibfile}

\providecommand{\bysame}{\leavevmode\hbox to3em{\hrulefill}\thinspace}
\providecommand{\MR}{\relax\ifhmode\unskip\space\fi MR }
\providecommand{\MRhref}[2]{%
  \href{http://www.ams.org/mathscinet-getitem?mr=#1}{#2}
}
\providecommand{\href}[2]{#2}
\begin{thebibliography}{1}

\bibitem{cochran1950}
William~G Cochran, \emph{The comparison of percentages in matched samples},
  Biometrika \textbf{37} (1950), no.~3/4, 256--266.

\bibitem{conover1998}
WJ~Conover, \emph{Practical nonparametric statistics}, John Wiley \& Sons,
  1998.

\bibitem{cronbach1955}
Lee~J Cronbach and Paul~E Meehl, \emph{Construct validity in psychological
  tests.}, Psychological bulletin \textbf{52} (1955), no.~4, 281.

\bibitem{haynes1995}
Stephen~N Haynes, David Richard, and Edward~S Kubany, \emph{Content validity in
  psychological assessment: A functional approach to concepts and methods.},
  Psychological assessment \textbf{7} (1995), no.~3, 238.

\bibitem{law1998}
Kenneth~S Law, Chi-Sum Wong, and William~M Mobley, \emph{Toward a taxonomy of
  multidimensional constructs}, Academy of management review \textbf{23}
  (1998), no.~4, 741--755.

\bibitem{lawshe1975}
Charles~H Lawshe, \emph{A quantitative approach to content validity}, Personnel
  psychology \textbf{28} (1975), no.~4, 563--575.

\bibitem{patil1975}
Kashinath~D Patil, \emph{Cochran's q test: Exact distribution}, Journal of the
  American Statistical Association \textbf{70} (1975), no.~349, 186--189.

\bibitem{R}
{R Core Team}, \emph{R: A language and environment for statistical computing},
  R Foundation for Statistical Computing, Vienna, Austria, 2017.

\end{thebibliography}

\end{document}